\documentclass[times]{elsarticle}
\usepackage{geometry} 
\usepackage{url}
\usepackage{rotating}
\usepackage{longtable}
\usepackage{listings}
\usepackage{etoolbox}
\usepackage{amssymb}
\usepackage{amsthm}

\newcommand{\pred}[1]{\mbox{\textit{#1}}}

\newcommand{\clpif}[1]{\mbox{:-}}
\usepackage{float}
\floatstyle{boxed}
\restylefloat{figure}
\newtheorem{theorem}{Theorem}
\newtheorem{lemma}{Lemma}
\newtheorem{definition}{Definition}

\newcommand{\ptab}{~~~~}
\newcommand{\pif}{\mbox{\textbf{if}}}
\newcommand{\pthen}{\mbox{\textbf{then}}}
\newcommand{\pelse}{\mbox{\textbf{else}}}
\newcommand{\pwhile}{\mbox{\textbf{while}}}
\newcommand{\pdo}{\mbox{\textbf{do}}}
\newcommand{\passign}{\mbox{:=}}
\newcommand{\ptrue}{\mbox{\textbf{true}}}
\newcommand{\pfalse}{\mbox{\textbf{false}}}
\newcommand{\WP}{\mbox{\textsc{wp}}}
\newcommand{\WLP}{\mbox{\textsc{wlp}}}

\title{Comparing Weakest Precondition and Weakest Liberal Precondition}

\author{Andrew E. Santosa}
\ead{dcsandr@nus.edu.sg}

\address{School of Computing, National University of Singapore}

\begin{document}

\begin{abstract}
  In this article we investigate the relationships between the
classical notions of weakest precondition and weakest liberal
precondition, and provide several results, namely that in general,
weakest liberal precondition is neither stronger nor weaker than
weakest precondition, however, given a deterministic and terminating
sequential while program and a postcondition, they are equivalent.
Hence, in such situation, it does not matter which definition is used.
\end{abstract}


\maketitle

\section{Introduction}
\label{sec:intro}

Recent years have seen the prevalent usage of \emph{symbolic
  execution\/}~\cite{king76symbolic} for program analysis. Typical
symbolic execution system builds \emph{path conditions\/}
corresponding to execution paths.  A path condition is a constraint
that represents logical relation between the input and output of an
execution path.  Its components are constraints modeling the input and
output relations of each program statement along the execution path. A
path condition can be used to determine the set of inputs that causes
a program to reach an error state. For example, given an array index
$i$ and array size bound $s,$ the path condition represents the
conditions on the input variables that makes array bounds violation
$i\geq s$ holding after executing a path. A constraint solver can be
used to compute an actual program inputs that causes the
violation. There are two well-known notions of the set of inputs
represented by a path and violation conditions: \emph{weakest
  precondition\/} and \emph{weakest liberal precondition}. These
notions are elements of the general notion of \emph{predicate
  transformation\/} introduced in \cite{dijkstra75gcl}. Whereas
weakest and weakest liberal preconditions computes input conditions in
a ``backward'' manner, in the literature, the notion of predicate
transformation also includes a ``forward'' transformation called
\emph{strongest postcondition}.

In this article, we explain how weakest and weakest liberal
preconditions are different. We also explain how that under a very
common condition of deterministic and terminating programs they are
equivalent. In Section \ref{sec:nonequivalence} we provide some
preliminary definitions together with our first result that in
general, weakest and weakest liberal preconditions are not equivalent.
We also present their relationships when the program is deterministic,
and when the program induces a satisfiable transition relation. In
Section \ref{sec:equivalence} we show that given a deterministic and
terminating while program, weakest and weakest liberal preconditions
are the same, and in Section \ref{sec:discussion} we show how to
define weakest liberal precondition in terms of weakest precondition,
and in Section \ref{sec:conclusion} we make some concluding remarks.

\section{Weakest and Weakest Liberal are Not Equivalent}
\label{sec:nonequivalence}

Here we clarify some terminologies. In this article, we adopt the more
common definition of weakest liberal precondition as in
\cite{dijkstra76lang}. However, in some
literature~\cite{bjornerinv97}, weakest liberal precondition is
instead termed weakest precondition. Our definition of weakest liberal
precondition is equivalent to the weakest precondition of
\cite{bjornerinv97}. On the other hand, weakest precondition that we
mean in this article is that of \cite{dijkstra75gcl} or
\cite{dijkstra76lang} which is also sometimes also termed
\emph{pre-image\/} in the literature (cf. the backward CTL decision
procedure in \cite{Huth:ModelChecking00}). Compared to weakest liberal
precondition, the notion of weakest precondition as in
\cite{dijkstra76lang} and \cite{dijkstra75gcl} adds the requirement
that the precondition should guarantee the termination of the
execution.

We now start with some formal definitions. We denote by $\tilde{x}$ a
sequence $x_0, \ldots, x_n$ of (program) variables with some
unspecified $n.$ We abuse the notion of program to also mean any of
its fragments such as, e.g., a statement is also a program. Now, any
program $P$ induces a \emph{transition relation\/} $\rho_P(\tilde{x},
\tilde{x}')$ on free variables $\tilde{x}$ and $\tilde{x}'$, where
$\tilde{x}$ represents the program variables before the transition and
$\tilde{x}'$ represents the program variables after the transition.
For example, an assignment statement $\tilde{x} ~\passign~
f(\tilde{x})$ induces the transition relation $\tilde{x}' =
f(\tilde{x}).$ In general, for any condition $\varphi,$ we write
$\varphi(\tilde{x})$ to clarify that $\tilde{x}$ and no other are the
free variables in $\varphi.$ Given a program $P$ and a postcondition
$\varphi(\tilde{x})$, the weakest liberal precondition of
$\varphi(\tilde{x})$ wrt. $P,$ written
$\WLP(P, \varphi(\tilde{x})),$ is the formula
\[
\forall \tilde{x}' : \rho_P(\tilde{x},\tilde{x}') \rightarrow \varphi(\tilde{x}')
\]
where $\varphi(\tilde{x}')$ is $\varphi(\tilde{x})$ with all
its free variables renamed to their primed versions. On the other hand,
the weakest precondition of $\varphi(\tilde{x})$ wrt. $P,$ written
$\WP(P, \varphi(\tilde{x})),$ is the formula
\[
\exists \tilde{x}' : \rho_P(\tilde{x},\tilde{x}') \wedge \varphi(\tilde{x}')
\]
We remove the subscript $P$ from the transition relation symbol
whenever it is clear from the context.

Weakest liberal precondition and weakest precondition are not
equivalent in general, as stated in the following theorem.
\begin{theorem} \label{theorem:neither}
In general, weakest liberal precondition is neither stronger nor weaker than weakest precondition.
\end{theorem}

\begin{proof}
If weakest precondition was stronger than weakest liberal
precondition, then the following would be unsatisfiable: 
\[
(\exists \tilde{x}' : \rho(\tilde{x},\tilde{x}') \wedge \varphi(\tilde{x}')) \wedge \neg(\forall \tilde{x}' :
\rho(\tilde{x},\tilde{x}') \rightarrow \varphi(\tilde{x}')). 
\]
This is equivalent to:
\[
(\exists \tilde{x}' : \rho(\tilde{x},\tilde{x}') \wedge \varphi(\tilde{x}')) \wedge (\exists \tilde{x}' : \rho(\tilde{x},\tilde{x}') \wedge \neg\varphi(\tilde{x}')).
\]
There exists some $\rho$ such that this formula is satisfiable, that is,
in case $\rho$ comes from nondeterministic statement. For example, when
$\rho(\tilde{x},\tilde{x}')$ is just a satisfiable constraint $\varphi(\tilde{x})$, which says nothing
about $\tilde{x}'$. More concrete example is when $\rho(\tilde{x},\tilde{x}')$ comes from the
statements \texttt{c=read();} or \texttt{c=rand(seed);} assuming the
\texttt{read()} and \texttt{rand(seed)} can return any value.

On the other hand, if weakest liberal precondition was stronger than
weakest precondition, then the following would be unsatisfiable:
\[
(\forall \tilde{x}' : \rho(\tilde{x},\tilde{x}') \rightarrow \varphi(\tilde{x}')) \wedge \neg(\exists
\tilde{x}' : \rho(\tilde{x},\tilde{x}') \wedge \varphi(\tilde{x}')). 
\]
This is equivalent to:
\[
(\forall \tilde{x}' : \rho(\tilde{x},\tilde{x}') \rightarrow \varphi(\tilde{x}')) \wedge (\forall \tilde{x}' : \rho(\tilde{x},\tilde{x}') \rightarrow \neg\varphi(\tilde{x}')). 
\]
However, also in this case there is some $\rho$ such that the formula is
satisfiable, that is, when $\rho$ is $\pfalse.$ A concrete example
of such $\rho$ is an \texttt{exit(0);} statement in C, or any other
statement that aborts the program.
\end{proof}

\section{Equivalence of Weakest and Weakest Liberal 
for Deterministic and Terminating While Programs}
\label{sec:equivalence}

\begin{theorem} \label{theorem:wpstronger}
  When the transition relation is deterministic, weakest precondition 
  is stronger than weakest liberal precondition.
\end{theorem}

\begin{proof}
  We can infer this from the proof of Theorem \ref{theorem:neither}
  above.  More formally, we show this by proving that the following is
  unsatisfiable when $\rho(\tilde{x},\tilde{x}')$ is $\tilde{x}' = f(\tilde{x})$ for some
  deterministic function $f$:
\[
(\exists \tilde{x}' : \rho(\tilde{x},\tilde{x}') \wedge \varphi(\tilde{x}')) \wedge \neg(\forall \tilde{x}' :
\rho(\tilde{x},\tilde{x}') \rightarrow \varphi(\tilde{x}')). 
\]
This is equivalent to:
\[
(\exists \tilde{x}' : \rho(\tilde{x},\tilde{x}') \wedge \varphi(\tilde{x}')) \wedge (\exists \tilde{x}' :
\rho(\tilde{x},\tilde{x}') \wedge \neg\varphi(\tilde{x}')). 
\]
Substituting $\rho(\tilde{x},\tilde{x}')$ with $\tilde{x}'=f(\tilde{x})$ we have: $\varphi(f(\tilde{x}))
\wedge \neg\varphi(f(\tilde{x}))$, 
which is unsatisfiable if f is deterministic. 
\end{proof}

\begin{theorem} \label{theorem:wlpstronger}
  When the transition relation is satisfiable, weakest liberal 
  precondition is stronger than weakest precondition. 
\end{theorem}

\begin{proof}
We can infer this from the proof of Theorem \ref{theorem:neither}
above. More formally, we proceed by showing a contradiction that
\begin{equation} \label{eqn:contradict2}
\WLP(P, \varphi(\tilde{x})) \not\rightarrow
\WP(P, \varphi(\tilde{x}))
\end{equation}
is unsatisfiable in case $\rho(\tilde{x},\tilde{x}')$ is satisfiable.
It is easy to see that (\ref{eqn:contradict2}) is equivalent to:
\[
\forall \tilde{x}' : \rho(\tilde{x},\tilde{x}') \rightarrow (\varphi(\tilde{x}') \wedge \neg\varphi(\tilde{x}'))
\]
which is absurd as $\varphi(\tilde{x}') \wedge
\neg\varphi(\tilde{x}')$ is $\pfalse$ and
$\rho(\tilde{x},\tilde{x}') \not\rightarrow \pfalse.$
\end{proof}

In the special case of sequential programs, since the weakest liberal
precondition is actually equivalent to weakest precondition. Following
is the proof why, for sequential programs, weakest liberal
precondition is equivalent to weakest precondition.

\begin{definition}
  A deterministic sequential while program may contain assignments, if
  conditionals, and while loops, and their sequential compositions in
  the usual manner. In addition, for any assignment $\tilde{x} ~\passign~ f(\tilde{x})$, $f$
  is a deterministic function.
\end{definition}

Let us now examine the transition relation induced by each of the
statement of a deterministic sequential while program:
\begin{enumerate}
\item For an assignment $\tilde{x} ~\passign~ f(\tilde{x})$, the transition relation $\rho(\tilde{x},\tilde{x}')$ is $\tilde{x}'=f(\tilde{x})$.
\item For an if conditional
\[
\pif ~ \varphi(\tilde{x}) ~\pthen ~P~ \pelse~ P'
\]
when the transition relation for $P$ is $\rho_P(\tilde{x},\tilde{x}')$ and the
transition relation for $P'$ is $\rho_{P'}(\tilde{x},\tilde{x}')$, the transition relation
$\rho(\tilde{x},\tilde{x}')$ induced by the if conditional is
\[
(\varphi(\tilde{x}) \wedge \rho_{P}(\tilde{x},\tilde{x}')) \vee (\neg \varphi(\tilde{x}) \wedge \rho_{P'}(\tilde{x},\tilde{x}'))
\]
\item For a while loop
\[
\pwhile ~ \varphi(\tilde{x})~ \pdo ~P
\]
when the transition relation for $P$ is $\rho_P(\tilde{x},\tilde{x}')$, then the
transition relation for the while loop is the infinite formula
\[
\begin{array}{c}
(\neg \varphi(\tilde{x}) \wedge \tilde{x}'=\tilde{x}) \vee\\
(\varphi(\tilde{x}) \wedge \rho_P(\tilde{x},\tilde{x}_1) \wedge \neg \varphi(\tilde{x}_1) \wedge \tilde{x}'=\tilde{x}_1) \vee\\
(\varphi(\tilde{x}) \wedge \rho_P(\tilde{x},\tilde{x}_1) \wedge \varphi(\tilde{x}_1) \wedge \rho_P(\tilde{x}_1,\tilde{x}_2) \wedge\neg
\varphi(\tilde{x}_2)\wedge \tilde{x}'=\tilde{x}_2) \vee\\
\ldots
(\varphi(\tilde{x}) \wedge \rho_P(\tilde{x},\tilde{x}_1) \wedge (\bigwedge_{i=2}^n : \varphi(\tilde{x}_{i-1}) \wedge
\rho_P(\tilde{x}_{i-1},\tilde{x}_i)) \wedge \neg \varphi(\tilde{x}_n) \wedge \tilde{x}'=\tilde{x}_n) \vee\\
\ldots
\end{array}
\]
or,
\[
\bigvee_{i=0}^{\infty} (\exists \tilde{x}_0,\ldots ,\tilde{x}_i :
(\bigwedge_{j=0}^{i-1} (\varphi(\tilde{x}_j) \wedge \rho_P(\tilde{x}_j,\tilde{x}_{j+1})) \wedge \neg \varphi(\tilde{x}_i) \wedge \tilde{x}'=\tilde{x}_i \wedge \tilde{x}=\tilde{x}_0))
\]
It is important to note here that for any nonterminating program $P$, $\neg
\varphi(\tilde{x}_i)$ for all $i$ is unsatisfiable, hence $\rho_P(\tilde{x},\tilde{x}')$ is $\pfalse.$
\end{enumerate}

Note that a deterministic while program induces a transition relation
that is always satisfiable, since if and while conditionals construct
two guarded program paths which guards are opposite of each
other. Hence, given a program execution state, both guards cannot be
unsatisfiable. Since a deterministic while program is both
deterministic and has transition relation that is always satisfiable,
Theorems \ref{theorem:wpstronger} and \ref{theorem:wlpstronger} seem
to have already suggested that a deterministic while program would
have equivalent weakest liberal precondition and weakest precondition,
however, here we will proceed more formally and carefully.
\begin{lemma} \label{lemma:equivalent}
The weakest liberal precondition of an assignment is equivalent to its weakest precondition.
\end{lemma}

\begin{proof}
  Given an assignment $\tilde{x} ~\passign~ f(\tilde{x})$ and a postcondition $\varphi(\tilde{x})$, the
  weakest liberal precondition is
\[
\forall \tilde{x}' : \tilde{x}'=f(\tilde{x}) \rightarrow \varphi(\tilde{x}')
\]
and the weakest precondition is 
\[
\exists \tilde{x}' : \tilde{x}'=f(\tilde{x}) \wedge \varphi(\tilde{x}'),
\]
each one is equivalent to $\varphi(f(\tilde{x})),$ given $f$ deterministic
function.
\end{proof}

\begin{lemma} \label{lemma:sequence}
  When for each program $P$ and $P',$ the weakest liberal
  precondition is equivalent to the weakest precondition given any
  postcondition, then given a postcondition $\varphi(\tilde{x}),$
  the sequence $P P'$ has equivalent weakest liberal precondition and
  weakest precondition.
\end{lemma}

\begin{proof}
  Given the postcondition $\varphi(\tilde{x}),$ the weakest
  liberal precondition of of $\varphi(\pred{x})$ wrt. $P'$ is
  $\pred{Pre}_{P'}$, which is necessarily equivalent to the weakest
  precondition of $\varphi(\pred{x})$ wrt. $P'$. Now, given
  $\pred{Pre}_{P'}$ as postcondition, the weakest liberal precondition
  and weakest precondition of $\pred{Pre}_{P'}$ wrt. $P$ are
  necessarily equivalent from our assumption that for any
  postcondition $\varphi$ and program $P,$ $\WLP(P,\varphi)
  \equiv \WP(P', \varphi).$
\end{proof}

\begin{theorem} \label{theorem:deterministic}
  Given a deterministic and terminating sequential while program $P$
  and a postcondition, the weakest liberal precondition of the program 
  wrt. the postcondition is equivalent to the weakest precondition of 
  the program wrt. the postcondition. 
\end{theorem}

\begin{proof}
  We prove inductively. When $P$ is just a sequence of assignments,
  from Lemma \ref{lemma:equivalent} and Lemma \ref{lemma:sequence} we
  obtain the desired result.

  Now let us assume $P$ to be an if conditional, say of the form
\[
\pif ~\varphi(\tilde{x})~ \pthen ~P~ \pelse ~P'
\]
As our induction hypothesis, we also assume that both $P$ and $P'$ have
equivalent weakest liberal precondition and weakest precondition given
any postcondition. Now suppose that the postcondition of the statement
is $\varphi.$ Recall that the transition relation of an if conditional is
\[
(\varphi(\tilde{x}) \wedge \rho_P(\tilde{x},\tilde{x}')) \vee (\neg \varphi(\tilde{x}) \wedge \rho_{P'}(\tilde{x},\tilde{x}'))
\]
The weakest liberal precondition of the if condition, given
$\varphi$ as postcondition is therefore
\[
(\forall \tilde{x}' : ((\varphi(\tilde{x}) \wedge \rho_{P}(\tilde{x},\tilde{x}')) \vee (\neg \varphi(\tilde{x}) \wedge
\rho_{P'}(\tilde{x},\tilde{x}'))) \rightarrow \varphi(\tilde{x}'))
\]
which is equivalent to
\[
(\varphi(\tilde{x}) \rightarrow (\forall \tilde{x}' : \rho_P(\tilde{x},\tilde{x}') \rightarrow
\varphi(\tilde{x}'))) \wedge (\neg \varphi(\tilde{x}) \rightarrow (\forall \tilde{x}' :
\rho_{P'}(\tilde{x},\tilde{x}') \rightarrow \varphi(\tilde{x}')))
\]
Note that in the above,
\[
\forall \tilde{x}' : \rho_P(\tilde{x},\tilde{x}') \rightarrow \varphi(\tilde{x}')
\]
and
\[
\forall \tilde{x}' : \rho_{P'}(\tilde{x},\tilde{x}') \rightarrow \varphi(\tilde{x}')
\]
are the weakest liberal preconditions of $\varphi(\tilde{x})$
wrt. respectively $P$ and $P'$. We name them
$\pred{Pre}_{P}(\tilde{x})$ and $\pred{Pre}_{P'}(\tilde{x})$,
respectively, obtaining (\ref{eqn:one}) below:
\begin{equation} \label{eqn:one}
(\varphi(\tilde{x}) \rightarrow \pred{Pre}_P(\tilde{x})) \wedge (\neg \varphi(\tilde{x}) \rightarrow \pred{Pre}_{P'}(\tilde{x}))
\end{equation}
Now the weakest precondition of $\varphi$ wrt. the $\pif$ condition, is:
\[
(\exists \tilde{x}': (\varphi(\tilde{x}) \wedge \rho_P(\tilde{x},\tilde{x}')) \vee (\neg \varphi(\tilde{x}) \wedge
\rho_{P'}(\tilde{x},\tilde{x}')) \wedge \varphi(\tilde{x}'))
\]
which is equivalent to
\[
(\varphi(\tilde{x}) \wedge (\exists \tilde{x}': \rho_P(\tilde{x},\tilde{x}') \wedge \varphi(\tilde{x}'))) \vee (\neg \varphi(\tilde{x}) \wedge (\exists \tilde{x}': \rho_{P'}(\tilde{x},\tilde{x}’) \wedge \varphi(\tilde{x}')))
\]
Since the weakest precondition and weakest liberal preconditions of
$P$ and $P'$ are equivalent, we get:
\[
(\varphi(\tilde{x}) \wedge \pred{Pre}_P(\tilde{x})) \vee (\neg \varphi(\tilde{x}) \wedge \pred{Pre}_{P'}(\tilde{x}))
\]
This is equivalent to (\ref{eqn:one}).

While loop of the syntax
\[
\pwhile ~\varphi(\tilde{x})~ \pdo ~P
\]
has the same semantics as the following infinite program consisting of
if conditionals.
\[
\begin{array}{l}
\pif ~\varphi(\tilde{x})~ \pthen\\
\ptab P\\
\ptab \pif ~\varphi(\tilde{x})~ \pthen\\
\ptab \ptab P\\
\ptab \ptab \ldots
\end{array}
\]
The infinite programs exactly induces the same transition relation as
the while loop presented above. Due to termination assumption, the
same while loop can be written using a finite number of if
conditionals (from the first if conditional up to the last (innermost)
if conditional where $\varphi(\tilde{x})$ is $\pfalse$). More importantly, the while
loop induces a transition relation that is satisfiable (not
$\pfalse$), that is, there is a possible execution from the point
before the loop to the point right after the loop. Since one if
conditional preserves the equivalence of weakest liberal precondition
and weakest precondition, as above, so does terminating while loops
(which are representable as finite number of ifs).
\end{proof}

\section{Discussion}
\label{sec:discussion}

It is easy to see that the following relationship holds between 
weakest liberal precondition and weakest precondition, where the 
weakest liberal precondition 
\[
\forall \tilde{x}' : \rho(\tilde{x},\tilde{x}') \rightarrow \varphi(\tilde{x}') 
\]
is actually equivalent to the negation of the weakest precondition of 
the negated postcondition. 
\[
\neg(\exists \tilde{x}' : \rho(\tilde{x},\tilde{x}') \wedge \neg\varphi(\tilde{x}')). 
\]
This fact has been mentioned by Bourdoncle in his abstract debugging 
approach \cite{bourdoncle93debug}, where he introduced two kinds of 
assertions to be guaranteed by a correctly running programs:
\emph{always\/} assertions and \emph{eventually\/} assertions. The 
proofs of both require program state-space exploration using backward 
fixpoint computations. The state-space exploration of the always 
assertions employ weakest liberal precondition while the state-space 
exploration of the eventually assertions employ weakest 
precondition. The intuitive relations between both assertions is that,
if suppose that we had an always assertion of some program correctness 
condition, and if the assertion holds, then in no circumstance that a 
program state where that assertion is violated can be eventually 
reached. That is, it is \emph{not\/} the case that a \emph{negation\/}
of the correctness condition eventually holds. 

Weakest precondition guarantees the total correctness of 
a Hoare's triples $\{\pred{Pre}\} ~ S ~ \{\varphi\}$, where 
$\pred{Pre}$ is a precondition, $\varphi$ a postcondition, and $S$
a statement. The notion of weakest liberal precondition, on the other 
hand, guarantees only partial correctness of the triples, where the 
postcondition is guaranteed to hold only when the statement was 
executed successfully. 

As a note, we can define weakest liberal precondition using weakest
precondition.  This does not mean, however, that we cannot implement
weakest liberal precondition propagation indirectly using weakest
precondition computation. Note that in a sequence $P P'$ the
weakest liberal precondition of a condition $\varphi(\tilde{x})$ wrt. the program $P'$ is
$\WLP(\varphi(\tilde{x}), P')$, which is equivalent to $\forall \tilde{x}'' :
(\rho_{P'}(\tilde{x},\tilde{x}'') \rightarrow \varphi(\tilde{x}''))$, where $\rho_{P'}$ is
the state transition relation defined by the program $P'.$ Now, the
weakest liberal precondition of the sequence is
\[
\forall \tilde{x}' : \rho_P(\tilde{x},\tilde{x}') \rightarrow (\forall \tilde{x}'' : (\rho_{P'}(\tilde{x}',\tilde{x}'') \rightarrow \varphi(\tilde{x}'')))
\]
which is equivalent to
\[
\forall \tilde{x}', \tilde{x}'' : (\rho_P(\tilde{x},\tilde{x}') \wedge \rho_{P'}(\tilde{x}',\tilde{x}'')) \rightarrow \varphi(\tilde{x}'').
\]
Notice that $\rho_P(\tilde{x},\tilde{x}') \wedge \rho_{P'}(\tilde{x}',\tilde{x}'')$ is $\WP(P P', \ptrue).$

\section{Concluding Remarks}
\label{sec:conclusion}

The semantics of the guarded commands language introduced in
\cite{dijkstra75gcl} embeds the notion of termination. In
\cite{dijkstra75gcl}, weakest precondition has to satisfy an
additional condition $Q$ (satisfiability of at least one guard in case
of guarded ifs, and a measure for the termination of a guarded loop),
which ensures the termination of the statement. However, $Q$ does not
exclude nondeterminism, and therefore from Theorems
\ref{theorem:neither}, \ref{theorem:wpstronger}, and
\ref{theorem:wlpstronger}, we infer that the notion of weakest
precondition and $Q$ in \cite{dijkstra76lang} is stronger than the
notion of weakest precondition used in this article.

We note that in this article, we have considered \emph{value\/}
nondeterminism of functions, while \cite{dijkstra75gcl} consider
\emph{control\/} nondeterminism where multiple guards can be true at
the same time and the semantics does not specify which branch is
taken. However, control nondeterminism can always be modeled using
value nondeterminism by having some guards which depend on random
value.

\end{document}